\documentclass[letterpaper, 10pt, conference]{IEEEtran} 
\makeatletter
\def\blfootnote{\xdef\@thefnmark{}\@footnotetext}
\makeatother
\normalsize

\DeclareMathAlphabet\mathbfcal{OMS}{cmsy}{b}{n}
\usepackage[colorlinks=true,pdfstartview=FitV,linkcolor=black,citecolor=black,urlcolor=blue,plainpages=false]{hyperref}

\usepackage{amsthm}
\usepackage{bbm}
\usepackage{algorithm}
\usepackage{algorithmic}
\usepackage{graphicx,amsmath,amssymb,latexsym,epsfig,url}

\usepackage{setspace}
\usepackage{psfrag}
\usepackage{array,booktabs,arydshln,xcolor,cite}

\usepackage{lscape}
\newtheorem{theorem}{Theorem}

\newtheorem{corollary}{Corollary}
\newtheorem*{proof of Theorem*}{Proof of Theorem 3}
\newtheorem{proof of Lemma}{Proof of Lemma}
\newtheorem{definition}{Definition}
\newtheorem{lemma}{Lemma}

\newcommand{\ignore}[1]{}
\begin{document}
\addtolength{\textheight}{0.8cm}
\vspace{-0.2 in}

\title{On the Trackability of  Stochastic Processes}
\author{\IEEEauthorblockN{Baran Tan Bacinoglu\IEEEauthorrefmark{1}, Yin Sun\IEEEauthorrefmark{3}, Elif Uysal\IEEEauthorrefmark{1}}
\IEEEauthorblockA{\IEEEauthorrefmark{1}METU, Ankara, Turkey,
\IEEEauthorrefmark{3}Auburn University, AL, USA,\\
 E-mail:  barantan@metu.edu.tr, yzs0078@auburn.edu, uelif@metu.edu.tr}
}

\bibliographystyle{IEEEtran}

\maketitle

\def\eg{\emph{e.g.}}
\def\ie{\emph{i.e.}}

\begin{abstract}
We consider the problem of tracking an unstable stochastic process $X_t$ by using causal knowledge of another stochastic process $Y_t$. We obtain necessary conditions and sufficient conditions for maintaining a finite tracking error. We provide necessary conditions as well as sufficient conditions for the success of this estimation, which is defined as order $m$ moment trackability. By-products of this study are connections between statistics such as R\'{e}nyi entropy, Gallager's reliability function, and the concept of anytime capacity.    
\end{abstract}
 
\section{Introduction}
\blfootnote{
This work has been supported in part by ONR grant N00014-17-1-2417, NSF grant CCF-1813050, TUBITAK Grant 
112E175 and Huawei.}
The tracking of unstable processes from noisy, delayed or infrequent samples is a fundamental problem that naturally arises in networked control systems. Non-stationary stochastic processes such as random walks and their scaling limits such as the Wiener process and the Ornstein -Uhlenbeck process are examples of unstable processes that are often used to model uncontrolled systems. Tracking of such processes may also arise in situations that do not require closed-loop control, and may have applications beyond networked control systems. For example, the Wiener process is used to model the physical diffusion process known as Brownian motion \cite{karatzasshrevefo1991}, option pricing in financial analysis \cite{karatzasshrevefo1998}, phase noise in communication channels \cite{21283} and forms a basis for analysis tools such as Feynman-Kac formula \cite{kac1979}. 

The  problem of communicating the state of unstable sources has been considered in \cite{Sahai2001AnytimeIT}, which introduced the notions of \emph{anytime reliability} and \emph{anytime capacity}. It was shown that anytime capacity provides the necessary and sufficient condition on the rate of an unstable scalar Markov source that can be tracked in the finite mean-squared error sense. Accordingly, it was claimed that anytime capacity, which is upper bounded by Shannon capacity, is the correct figure of merit to measure the quality of a channel on the purpose of tracking an unstable source and also controlling through an unreliable channel \cite{7500143}. However, while anytime capacity is known to be strictly positive for some channels, a closed-form expression for anytime capacity has not been shown as opposed to Shannon capacity which can be expressed as an optimization of mutual information. On the other hand, anytime capacity of particular channels such as erasure channels with feedback \cite{Simsek2004AnytimeIT} and Markov channels \cite{7500143} have been derived.     

Aside from the theory developed in \cite{Sahai2001AnytimeIT}, the problem of stabilizing a system with limited communication has been extensively studied from stochastic control \cite{BansalBasar1989, ImerYukselandBasar2006, 5290272, yuksel10,yuksel2012}, rate-distortion theory \cite{DBLP:conf/cdc/CharalambousKS09, 8395023, 8693967, 8437916, stavrou2019sequential} and joint source-channel coding \cite{8693975}  perspectives for linear systems and from the perspectives of metric and topological entropy for non-linear dynamical systems \cite{KawanYuksel2019}. In our study, the state estimation side of this problem is investigated from a perspective that is centered on 
a definition of reliable estimation which we refer as \emph{order $m$ moment trackability} in accordance with \emph{$m$-th moment stability}. Based on this definition, we study the estimation of integer-valued \footnote{Our results are for integer-valued sources, however, note that this is not restrictive for digital systems where data is represented using integers.} stochastic processes which may represent linear or non-linear discrete-time systems.

Our contributions are as follows:
\begin{itemize}
    \item We show two moment-entropy inequalities for integer-valued  random variables inspired from the inequality for the moments of guessing random variables in \cite{Arikan1996}. One of these bounds is for bounded integer-valued random variables (see Lemma \ref{momententropyforboundedintegers}) while the other (see Lemma \ref{momententropyforintegers}) is valid for integer-valued random variables that do not necessarily have finite support.
    \item We provide necessary conditions (see Theorem \ref{necessaryconditionfordiscrete} and Theorem \ref{necessaryconditionforintegers}) for tracking integer-valued sources using causal information. Corollaries of Theorem \ref{necessaryconditionfordiscrete}   are upper bounds on anytime capacity based on Gallager's reliability function and the Gartner-Ellis limit of the information density between channel inputs and outputs.
    \item We provide sufficient conditions   for tracking integer-valued sources using causal information in Theorem \ref{sufficientconditionfordiscrete} and Theorem \ref{sufficientconditionforrhoestimators} where the former is based on an upper bound for the estimation error of maximum a posteriori (MAP) estimators (Lemma \ref{MAPboundforgeneraldistance}) and the latter is based on estimators we suggest.
\end{itemize}


\section{System Model} 

Consider the problem of tracking a scalar discrete-time  and discrete-valued stochastic process $\{X_t\}_{t=1,2,\ldots}$ based on causal knowledge of another stochastic process $\{Y_t\}_{t=1,2,\ldots}$. 
At any time $t$, the estimator generates a guess $\hat X_t = f_t(Y_{1:t}) $ of the current value $X_t$, where $f_t(\cdot) $ is a function and $Y_{1:t}= (Y_1, Y_2,\ldots, Y_t)$ is the   information that is available at time $t$.

\begin{definition}
For any $m>0$, $\{X_t\}_{t=1,2,\ldots}$ is said to be {order $m$ moment trackable} based on $\{Y_t\}_{t=1,2,\ldots}$ if there exists a family of  functions $\{f_t(\cdot)\}_{t=1,2,...}$ such that $\hat X_t = f_t(Y_{1:t})$ and 
\begin{equation}
\sup_{t>0} \mathbb{E}\left[ \vert X_{t}-\hat{X}_{t} \vert^{m}\right] < \infty. 
\end{equation}
\end{definition}

%
The first goal of the present paper is to find necessary conditions and sufficient conditions for the $m$-th moment trackability of process $\{X_t\}_{t=1,2,\ldots}$ based on the side information process $\{Y_t\}_{t=1,2,\ldots}$. In \cite{Sahai2001AnytimeIT}, the \emph{anytime capacity} of a noisy channel was shown to be a necessary and sufficient quality measure of a channel to allow order $m$ moment trackability of a Markov source $\{S_t\}_{t=1,2,\ldots}$ based on the channel output $\{Y_t\}_{t=1,2,\ldots}$,
The second goal of the  paper is to find new bounds of the anytime capacity, based on the trackability results. 





\section{Main Results}
\ignore{

Consider a scalar discrete-time stochastic process $\lbrace X_{t}\rbrace_{t>0}$ (source) and let $\lbrace Y_{t}\rbrace_{t>0}$ (causal information) be another discrete-time stochastic process from which an estimation of $ X_{t}$ is obtained for $t>0$. We will consider the conditions for the trackability of $\lbrace X_{t} \rbrace$ through $\lbrace Y_{t} \rbrace$ where we define trackability as follows:
\begin{definition}
For $m$ being a real number, a process $\lbrace X_{t} \rbrace$ is said to be order $m$ moment trackable though process $\lbrace Y_{t} \rbrace$ if there exists a family of estimators $\lbrace \hat{X}_{t} \rbrace$ for $t > 0$ such that:
\begin{equation}
\sup_{t>0} \mathbb{E}\left[ \vert X_{t}-\hat{X}_{t} \vert^{m}\right] < \infty, 
\end{equation}
where $\hat{X}_{t}$ is a function of $Y_{1:t}$, i.e., $\hat{X}_{t}=\hat{X}_{t}(Y_{1:t})$ for $Y_{1:t}$ denotes $\lbrace Y_{1},Y_{2},\ldots,Y_{t}\rbrace$.
\end{definition}
Our definition of trackability is based on the definition of tracking in \cite{Sahai2001AnytimeIT} and is also consistent with finite-moment stability condition used in the literature. In our definition, the trackability depends on $\lbrace Y_{t} \rbrace$ from which the estimation $\hat{X}_{t}(Y_{1:t})$, that actually tracks $X_{t}$, is generated.
}
\subsection{Necessary Conditions for Trackability}
We  provide two necessary conditions for order $m$ moment trackability, which are expressed in terms of R\'{e}nyi entropy and  information density. 
The R\'{e}nyi entropy of order $\alpha$, where $\alpha \geq 0$  and $\alpha \neq 1$, is defined as \cite{renyi1961}
\begin{align}
H_{\alpha}(X)&=\frac{1}{1-\alpha}\log\mathbb{E}\left[ P_{X}(X)^{\alpha-1}\right] \\
& = \frac{1}{1-\alpha}\log\left[\displaystyle\sum_{x\in \mathcal{X}} P_{X}(x)^{\alpha}\right].
\end{align}
Given joint distribution $P_{XY}$, the information density function is defined as \cite{pinsker1964information}
\begin{align}
i(x;y)= & \log\left[ \frac{P_{XY}(x,y)}{P_{X}(x)P_{Y}(y)}\right]. \label{eq_density_discrete}
\end{align}
%
%

The first necessary condition that we present is as follows:

\begin{theorem}
\label{necessaryconditionfordiscrete}
If $\{X_t\}_{t=1,2,\ldots}$ is an integer-valued stochastic process that satisfies 
\begin{align}\label{eq_c}
&\vert X_{t} \vert \leq c_{t}, \\
&\lim_{t \rightarrow \infty}\frac{1}{t}\log(\log(c_{t})) =0,\label{eq_logc}
\end{align}
then $\{X_t\}_{t=1,2,\ldots}$ is order $m$ moment trackable based on $\lbrace Y_{t} \rbrace_{t=1,2,\ldots}$, where $Y_{t} \in \mathcal{Y}$ and $\vert \mathcal{Y} \vert < \infty$, only if the following inequality holds, for all $\rho\in(0,m]$ and $q>\rho+1$, 
\begin{align}\label{eq_necessaryconditionfordiscrete}
&\liminf_{t \rightarrow \infty}-\frac{1}{\rho t}\log\mathbb{E}\left[\mathbb{E}\left[  e^{-\frac{\rho}{q}i(X_{t};Y_{1:t})} \Big| Y_{1:t}\right]^{q}\right]  \nonumber\\
\geq&\limsup_{t \rightarrow \infty} \frac{1}{t}H_{\frac{q-1}{q-\rho-1}}(X_{t}).
\end{align}
\end{theorem}
\begin{proof}
See Appendix \ref{necessaryconditionfordiscrete:proof}.
\end{proof}

The proof of Theorem \ref{necessaryconditionfordiscrete} uses the following moment-entropy inequality for the R\'{e}nyi entropy, which is inspired by Theorem 1 in \cite{Arikan1996}.
\begin{lemma}
\label{momententropyforboundedintegers}
If $X$ is an integer-valued random variable taking values from the set $\mathcal{X}=\lbrace -M_{-}, \ldots,-1,0,1,\ldots, M_{+}\rbrace$ where $M_{-}$ and $M_{+}$ are positive integers, then
for all $\rho \geq 0$ 
\begin{eqnarray}
\label{absmomentandRenyi}
&&\mathbb{E}[\vert X \vert^{\rho}]+ 1 \geq \left[ 3+\log(M_{-}M_{+})\right]^{-\rho}e^{\rho H_{\frac{1}{1+\rho}}(X)}.
\end{eqnarray}

\end{lemma}
\begin{proof}
See Appendix \ref{momententropyforboundedintegers:proof}.
\end{proof}
\ignore{
Observe that, the left hand side (LHS) of (\ref{necessaryconditionforGallagerRenyi}) can be written in terms of \emph{information density} function which is defined as follows:
\begin{definition}
Given joint distribution $P_{XY}$ for discrete random variables $X$ and $Y$,  the information density function is defined as follows:
\begin{equation}
i_{P_{XY}}(x;y):= \log\left( \frac{P_{XY}(x,y)}{P_{X}(x)P_{Y}(y)}\right)=\log\left( \frac{P_{X\vert Y}(x\vert y)}{P_{X}(x)}\right). 
\end{equation}
\end{definition}
In order to simply the notation, we drop the subscript $P_{XY}$ and denote $i_{P_{XY}}(\cdot;\cdot)$ as $i(\cdot;\cdot)$ assuming that the associated joint density can be inferred from the context. Accordingly, in case $Y_{1:t}$ is a discrete random variable, one can express (\ref{necessaryconditionforGallagerRenyi}) as follows:
\begin{eqnarray}
&&\liminf_{t \rightarrow \infty}-\frac{1}{mt}\log\mathbb{E}\left[\mathbb{E}\left[  e^{-\frac{m}{p(m+1)}i(X_{t};Y_{1:t})} \mid Y_{1:t}\right]^{p(m+1)}\right]  \nonumber\\
&&\geq\limsup_{t \rightarrow \infty} \frac{1}{t}H_{\frac{p(m+1)-1}{(p-1)(m+1)}}(X_{t}),
\end{eqnarray}
where $i(X_{t};Y_{1:t})$ is the information density random variable for joint density $P_{X_{t}Y_{1:t}}$.
}
Lemma  \ref{momententropyforboundedintegers} requires that $M_{-}$ and $M_{+}$ are finite. As a result, Theorem \ref{necessaryconditionfordiscrete} only applies to stochastic processes that satisfy \eqref{eq_c} and \eqref{eq_logc}. Next, we will provide a necessary condition for the  trackability of unbounded stochastic processes in Theorem \ref{necessaryconditionforintegers}, which is based on the following moment-entropy inequality.
\begin{lemma}
\label{momententropyforintegers}
If $X$ is an integer-valued random variable, then for all $\rho \in (0, m) $
\begin{align}
\label{absmomentandRenyiforintegers}
\mathbb{E}[\vert X \vert^{m}]  + 1 \geq \left[1+2\zeta\left(\frac{m}{\rho}\right)\right]^{-\rho}e^{\rho  H_{\frac{1}{1+\rho}}(X)},
\end{align}
where $\zeta(\cdot)$ is the Riemann zeta function
\begin{align}
\zeta(s) = \sum_{n=1}^\infty \frac{1}{n^s}.
\end{align}
\end{lemma}

\begin{proof}
See Appendix \ref{momententropyforintegers:proof}.
\end{proof}
\ignore{
Using Lemma \ref{momententropyforintegers}, we get another necessary condition for order $m$ moment trackability, which does not require $X_t$ to be bounded. 
}
\begin{theorem}
\label{necessaryconditionforintegers}
An integer-valued stochastic process  $\{X_t\}_{t=1,2,\ldots}$ is order $m$ moment trackable based on $\lbrace Y_{t} \rbrace_{t=1,2,\ldots}$ , where $Y_{t} \in \mathcal{Y}$ and $\vert \mathcal{Y} \vert < \infty$, only if \eqref{eq_necessaryconditionfordiscrete} holds  for all $\rho\in(0,m)$ and $q>\rho+1$. 
\end{theorem}
\begin{proof}
The proof is identical to the proof of Theorem \ref{necessaryconditionfordiscrete}, except that it uses Lemma \ref{momententropyforintegers} instead of Lemma \ref{momententropyforboundedintegers}. Note that $\zeta(\frac{m}{\rho})$ is finite for all $\rho\in(0,m)$.
\end{proof}
Theorem \ref{necessaryconditionforintegers} requires a weaker condition than Theorem \ref{necessaryconditionfordiscrete}. Accordingly,  the result of Theorem \ref{necessaryconditionforintegers} is weaker than that of Theorem \ref{necessaryconditionfordiscrete}. For this, notice that $\rho=m$ is not allowed in Theorem \ref{necessaryconditionforintegers}.

\subsection{Upper Bounds of Anytime Capacity}
Now, we show that (\ref{eq_necessaryconditionfordiscrete}) implies two inequalities that provide upper bounds on anytime capacity. First one can be expressed in terms of Gallager's reliability function which is defined as \cite{Gallager1965}
\begin{align}
&E_{0}(\rho,P_{Y\vert X}, P_{X}) \nonumber\\
=&-\log \sum_{y\in \mathcal{Y}} \left( \sum_{x\in \mathcal{X}}P_{X}(x)[P_{Y\vert X}(y \vert x)]^{\frac{1}{1+\rho}}\right)^{1+\rho}.
\end{align}
In \cite{5707067}, an alternative expression for Gallager's reliability function was used as follows
\begin{align}
\label{Gallagerfuntionwithinformationdensity}
&E_{0}(\rho,P_{Y\vert X}, P_{X}) =-\log\mathbb{E}\left[\mathbb{E}\left[  e^{-\frac{1}{1+\rho}i(\bar{X};Y)} \Big| Y\right]^{1+\rho}\right],
\end{align}
where $P_{X Y\bar{X}}(x,y,\bar{x})=P_{X}(x)P_{Y\vert X}(y\vert x)P_{X}(\bar{x})$ is the joint density for  $X$, $Y$ and $\bar{X}$.

In this paper, we find the following expression of Gallager’s reliability
function convenient, due to its connection with the LHS of (\ref{eq_necessaryconditionfordiscrete}). 
\begin{align}
\label{Gallagerfuntionwithinformationdensity2}
&E_{0}(\rho,P_{Y\vert X}, P_{X}) =-\log\mathbb{E}\left[\mathbb{E}\left[  e^{-\frac{\rho}{1+\rho}i(X;Y)} \Big| Y\right]^{1+\rho}\right].
\end{align}
Using (\ref{Gallagerfuntionwithinformationdensity2}), one can observe that the LHS of (\ref{eq_necessaryconditionfordiscrete}) becomes the Gallager's reliability function as $q$ reduces to $\rho+1$. Based on this observation, we derive the following corollary of Theorem \ref{necessaryconditionfordiscrete}:
\begin{corollary}
\label{necessaryconditionfordiscretecorollary2}
\label{necessaryconditionfordiscrete2}
Suppose that $S_{t}\rightarrow X_{1:t} \rightarrow Y_{1:t}$ is  a Markov chain for each $t$. If $\{S_t\}_{t=1,2,\ldots}$ is an integer-valued stochastic process that satisfies 
\begin{align}\label{eq_c1}
&\vert S_{t} \vert \leq c_{t}, \\
&\lim_{t \rightarrow \infty}\frac{1}{t}\log(\log(c_{t})) =0,\label{eq_logc2}
\end{align}
then $\{S_t\}_{t=1,2,\ldots}$ is order $m$ moment trackable based on $\lbrace Y_{t} \rbrace_{t=1,2,\ldots}$, where $Y_{t} \in \mathcal{Y}$ and $\vert \mathcal{Y} \vert < \infty$, only if
\begin{equation}
\label{necessaryconditionforGallagerRenyicorollary2}
\liminf_{t \rightarrow \infty}\frac{1}{mt}E_{0}(m,P_{Y_{1:t}\vert X_{1:t}}, P_{X_{1:t}})\! \geq \limsup_{t \rightarrow \infty} \frac{1}{t}H_{\infty}(S_{t}).\!\!\!
\end{equation}

\end{corollary}
\begin{proof}
Apply Theorem \ref{necessaryconditionfordiscrete} for $S_{t}$ considering $\rho=m$ and the limit that $q$ reduces to $\rho+1$ yields
\begin{equation}
\label{limitingboundforGallager}
\liminf_{t \rightarrow \infty}\frac{1}{mt}E_{0}(m,P_{Y_{1:t}\vert S_{t}},P_{S_{t}}) \geq\limsup_{t \rightarrow \infty} \frac{1}{t}H_{\infty}(S_{t}).
\end{equation}
Observe that (\ref{limitingboundforGallager}) implies (\ref{necessaryconditionforGallagerRenyicorollary2}) as $E_{0}(m,P_{Y_{1:t}\vert S_{t}},P_{S_{t}})$ is upper bounded by $E_{0}(m,P_{Y_{1:t}\vert X_{1:t}},P_{X_{1:t}})$ due to data-processing inequality for R\'{e}nyi divergence (see  \cite[Theorem 5]{5707067}).
\end{proof}
Corollary \ref{necessaryconditionfordiscretecorollary2} can be related to the $\alpha$-anytime capacity of a channel  (see \cite[Definition 3.2]{1661825}) when we consider the  following communication system. Let $Y_{1:t}$ be the outputs of a channel given by $P_{Y_{1:t}\vert X_{1:t}}(y_{1:t}\vert x_{1:t})$ with  $X_{1:t}$ being inputs that encode a source $\lbrace S_{t} \rbrace$ \footnote{A causal and general communication system as such is given by $P_{Y_{1:t}\vert X_{1:t}}(y_{1:t}\vert x_{1:t})=\prod_{t'=1}^{t}P_{Y_{t'}\vert X_{1:t'}, Y_{1:t'-1}}(y_{t'}\vert x_{1:t'}, y_{1:t'-1})$ and $P_{X_{1:t}\vert S_{1:t}}(x_{1:t}\vert s_{1:t})=\prod_{t'=1}^{t}P_{X_{t'}\vert S_{1:t'}}(x_{t'}\vert s_{1:t'})$ if we describe the encoding in terms of conditional probabilities for ease of description.}. As the outputs of the channel depend on the source process only through the channel inputs, the system follows $S_{t}\rightarrow X_{1:t} \rightarrow Y_{1:t}$. For ease of analysis, we will consider the type of source representing a stream of bits with fixed rate as follows:
\begin{definition}
For $R$ being a positive integer, a discrete-time process  $\lbrace S_{t} \rbrace$ is said to be a rate-$R$ source if it obeys:
\begin{equation}
S_{t+1}=2^{R}S_{t}+W_{t},
\end{equation}
where $\lbrace W_{t}\rbrace$ is an i.i.d. process  such that $W_{t}$ is uniformly chosen from the set $\left\lbrace 0,1,...., 2^{R}-1\right\rbrace $, and $X_{0}=0$.
\end{definition}
Note that a rate-$R$ source satisfies $ \vert S_{t} \vert \leq 2^{Rt}$ almost surely 
and $H_{\infty}(S_{t})=Rt\log(2)$ as it has a uniform distribution for all $t$. Accordingly, we can apply Corollary \ref{necessaryconditionfordiscretecorollary2} to a rate-$R$ source and show 
the following
\begin{corollary}
\label{upperbounforanytimecapacityusingGallagercorollarly}
If $C_{\text{any}}(\alpha)$ is the $\alpha$-anytime capacity of a discrete memoryless channel (DMC) without feedback, $R$ is an positive integer, $m>0$ is an arbitrary positive number, and 
\begin{equation}
\label{anycapacityfortracking}
R\log(2) \leq C_{\text{any}}(mR), 
\end{equation}
then
\begin{equation}
\label{upperbounforanytimecapacityusingGallager}
R\log(2) \leq \frac{E_{0}(m)}{m},
\end{equation}
where $E_{0}(m)=\sup_{P_{X}} E_{0}(m,P_{Y\vert X},P_{X})$ for given  transition probabilities $P_{Y\vert X}$ of the channel.
\end{corollary}
\begin{proof}
First suppose that (\ref{anycapacityfortracking}) holds which means a rate-$R$ source is order $m$ moment trackable though a DMC with anytime capacity $C_{any}(\alpha)$ (see \cite[Theorem 3.3]{1661825}). On the other hand, if a rate-$R$ source is order $m$ moment trackable through a DMC, the following should also hold:
\begin{equation}
\label{trackabilityboundforGallagergeneral}
\liminf_{t \rightarrow \infty}\frac{1}{mt}E_{0}(m,P_{Y_{1:t}\vert X_{1:t}}, P_{X_{1:t}}) \geq R\log(2),
\end{equation}
which follows from Corollary \ref{necessaryconditionfordiscretecorollary2}. Moreover, this implies
(\ref{upperbounforanytimecapacityusingGallager}) as $E_{0}(m,P_{Y_{1:t}\vert X_{1:t}}, P_{X_{1:t}})\leq tE_{0}(m)$ (see  \cite[Theorem 5]{Gallager1965}) for DMCs without feedback. 
\end{proof}
A result that is similar to Corollary \ref{upperbounforanytimecapacityusingGallagercorollarly} was shown (see \cite[Theorem 3.3.2 ]{Simsek2004AnytimeIT}) for  symmetric DMCs with feedback based on sphere packing exponent. On the other hand, Corollary \ref{upperbounforanytimecapacityusingGallagercorollarly} holds both for asymmetric and symmetric DMCs without feedback.

The second inequality that we provide can be obtained \footnote{See Appendix  \ref{proof:GartnerEllisfortrackinginformationdensity} for the proof.} from (\ref{necessaryconditionforGallagerRenyicorollary2}) while considering a rate-$R$ source for $S_{t}$. Accordingly, when $Y_{1:t}$ are the outputs of a channel with inputs $X_{1:t}$ that encode a rate-$R$ source, the source is order $m$ trackable based on $Y_{1:t}$ only if:  
\begin{equation}
\label{GartnerEllisfortrackinginformationdensity}
\liminf_{t \rightarrow \infty}\frac{1}{\rho t}\log\mathbb{E}\left[ e^{ \rho i(X_{1:t};Y_{1:t})} \right]
\geq R\log(2).
\end{equation}

In fact, the LHS of (\ref{GartnerEllisfortrackinginformationdensity}) is the  Gartner-Ellis limit of  $i(X_{1:t};Y_{1:t})$ which provides another upper bound for anytime capacity if we use (\ref{GartnerEllisfortrackinginformationdensity}) instead of (\ref{trackabilityboundforGallagergeneral}) in the proof of Corollary \ref{upperbounforanytimecapacityusingGallagercorollarly}. Also, observe that both  (\ref{trackabilityboundforGallagergeneral}) and (\ref{GartnerEllisfortrackinginformationdensity}) can be applied for channels other than DMCs without feedback.

\subsection{Sufficient Conditions for Trackability}
Next, we provide two sufficient conditions for order $m$ moment trackability. The first one is based on MAP estimators.
\begin{definition}
An estimator $\hat{X}_{t}^{(\text{MAP})}$ is said to be a maximum a posteriori (MAP) estimator if 
\begin{equation}
\hat{X}_{t}^{(\text{MAP})}=\arg\max_{x \in \mathcal{X}} P_{X_{t}\vert Y_{1:t}}(x\vert Y_{1:t}),
\end{equation}  
with ties in the maximization  broken arbitrarily.
\end{definition} 
We will use the following lemma to derive a sufficient condition for order $m$ moment trackability  based on MAP estimators:
\begin{lemma}
\label{MAPboundforgeneraldistance}
For an integer-valued stochastic process $\lbrace X_{t} \rbrace$, a discrete-valued stochastic process $\lbrace Y_{t} \rbrace $ and  $d(\cdot,\cdot)$ being a distance metric such that $d: \mathbb{Z}\times \mathbb{Z} \rightarrow \mathbb{Z}^{\geq 0}$ we have the following for arbitrary real numbers  $\rho > 0$ and $s >1$:
\begin{align}
\label{upperboundforMAPerrorongeneraldistance}
&\mathbb{E}\left[ d(X_{t},\hat{X}_{t}^{(\text{MAP})}) \right]\leq \zeta(s)\sum_{y_{1:t}}P_{Y_{1:t}}(y_{1:t})\sum_{x} \nonumber\\
& [P_{X_{t}\vert Y_{1:t}}(x\vert y_{1:t})]^{\frac{1}{\rho + 1}}\left[ \sum_{x'}[P_{X_{t}\vert Y_{1:t}}(x'\vert y_{1:t})]^{\frac{1}{\rho + 1}}d(x,x')^{\frac{s}{\rho}}\right]^{\rho},\nonumber\\
\end{align}
\end{lemma}
\begin{proof}
See Appendix \ref{MAPboundforgeneraldistance:proof}
\end{proof}
A sufficient condition for order $m$ moment trackability using Lemma \ref{MAPboundforgeneraldistance} is as follows:
\begin{theorem}
\label{sufficientconditionfordiscrete}
Let 
\begin{align}
&\tau(x,y_{1:t}) \nonumber\\
=&\mathbb{E}\left[ P_{X_{t}\vert Y_{1:t}}(X_{t}\vert Y_{1:t})^{-\frac{m}{m+1}} \vert X_{t}-x \vert^{s} \mid Y_{1:t}=y_{1:t}\right],
\end{align}
where $s > 1$ is an arbitrary real number and $m$ is an integer. Then, the integer-valued stochastic process $\{X_t\}_{t=1,2,\ldots}$ is order $m$ moment trackable using $\lbrace Y_{t} \rbrace_{t=1,2,\ldots}$ if
\begin{equation}
\label{sufficientconditionforRenyi}
\sup_{t>0}\mathbb{E}\left[ \tau(X_{t}, Y_{1:t})^{m}P_{X_{t}\vert Y_{1:t}}(X_{t}\vert Y_{1:t})^{-\frac{m}{m+1}} \right] <\infty.
\end{equation}
\end{theorem}
\begin{proof}
Apply Lemma \ref{MAPboundforgeneraldistance} for $d(x,x')=\vert x-x' \vert^{m}$ and $\rho=m$, then observe that 
\begin{eqnarray}
\label{upperboundforMAPerror}
&&\mathbb{E}\left[ \big\lvert X_{t}- \hat{X}_{t}^{(\text{MAP})}\big\rvert^{m} \right] \leq \nonumber\\
&&\zeta(s) \mathbb{E}\left[ \tau(X_{t}, Y_{1:t})^{m}P_{X_{t}\vert Y_{1:t}}(X_{t}\vert Y_{1:t})^{-\frac{m}{m+1}} \right].
\end{eqnarray}
\end{proof}
In addition to MAP estimators, we consider another type of estimators which are defined below:
\begin{definition}
\label{rhoestimators}
For $\rho>0$ being an arbitrary real number, let $\lbrace \hat{X}_{t}^{(\rho)}(Y_{1:t})\rbrace$ be a family of estimators such that $\hat{X}_{t}^{(\rho)}(y_{1:t})$ is  uniformly chosen from the set $\mathcal{A}_{t}(\rho,y_{1:t},J_{t}(\rho,y_{1:t}))$ where
\begin{eqnarray}
&&\mathcal{A}_{t}(\rho,y_{1:t},c)=\nonumber\\
&&\Big\{ x: \frac{P_{X_{t}\vert Y_{1:t}}(x\vert y_{1:t})}{P_{X_{t}\vert Y_{1:t}}(x'\vert y_{1:t})} \geq c \vert x-x'\vert^{\rho} ,\forall x'\Big\}
\end{eqnarray}
and
\begin{equation}
\label{Jdefinition}
J_{t}(\rho,y_{1:t})=\sup\lbrace c \geq 0: \mathcal{A}_{t}(\rho,y_{1:t},c) \neq \emptyset\rbrace.
\end{equation}
\end{definition}
Observe that, as opposed to MAP estimators,  the estimator $\hat{X}_{t}^{(\rho)}$ has a notion of distance and it requires that a possible value to be less likely proportional with its distance to the estimate. This requirement is natural as more likely values cluster around the estimate value.
Accordingly, considering the family of estimators $\lbrace \hat{X}_{t}^{(\rho)}\rbrace$ yields
\begin{theorem}
\label{sufficientconditionforrhoestimators}
If $p>1$  and $s > 1$ are arbitrary real numbers, and $m$ is a positive integer, then, the integer-valued stochastic process $\{X_t\}_{t=1,2,\ldots}$ is  order $m$ moment trackable  based on $\lbrace Y_{t} \rbrace_{t=1,2,\ldots}$ if
\begin{eqnarray}
\label{sufficientconditionforJt}
&&\sup_{t>0}\mathbb{E}\left[ \mathbb{E}\left[ P_{X_{t}\vert Y_{1:t}}(X_{t}\vert Y_{1:t})^{-\frac{m}{m+1}}\mid  Y_{1:t}\right]^{p(m+1)}\right]^{\frac{1}{p}}\nonumber\\
&&\times\mathbb{E}\left[ J_{t}(sm(m+1),Y_{1:t})^{\frac{p}{1-p}}\right]^{\frac{p-1}{p}}< \infty,\nonumber\\
\end{eqnarray}
where $J_{t}(\rho,y_{1:t})$ is as defined in (\ref{Jdefinition}) . 
\end{theorem}
\begin{proof}
See Appendix \ref{proof:sufficientconditionforrhoestimators}.
\end{proof}
Note that the first term in (\ref{sufficientconditionforJt}) can be expressed in terms of conditional R\'{e}nyi \footnote{We consider the definition of conditional R\'{e}nyi entropy that fits our case.} entropy when $p=1$ while $J_{t}$ function in the second term can be considered as a measure for the shape of the conditional distribution $P_{X_{t}\vert Y_{1:t}}(x\vert Y_{1:t})$.
\section{Conclusion}
We considered necessary and sufficient conditions for tracking a random source. Our results may provide insights to the design of causal information (via real-time coding) for systems that rely on the tracking of random sources.   

\section{Acknowledgements}
The authors would like to thank Serdar Yuksel and Kemal Leblebicioglu for their valuable comments.
\appendix
\subsection{The Proof of Theorem \ref{necessaryconditionfordiscrete}}
\label{necessaryconditionfordiscrete:proof}
Consider arbitrary estimators $\lbrace \hat{X}_{t} \rbrace$ such that $ \vert \hat{X}_{t} \vert \leq c_{t}$  for $t > 0$ \footnote{Clearly, any estimator $ \hat{X}_{t}$ which can take values that are outside of $[-c_t,c_t]$ is suboptimal for minimizing $\mathbb{E}\left[ \vert X_{t}- \hat{X}_{t}\vert^{m} \right]$.} . Let us define estimators $\lbrace \hat{X}_{t}^{(c)} \rbrace$ such that $\hat{X}_{t}^{(c)}= \lceil\hat{X}_{t}\rceil$ where $\lceil\cdot\rceil$ is the ceiling function. If $m \in (1,\infty)$, 
\begin{align}
\label{Minkowskiinequalityforquantization}
\mathbb{E}\left[ \vert X_{t}- \hat{X}_{t}^{(c)}\vert^{m} \right]^{\frac{1}{m}}
\leq \mathbb{E}\left[ \vert X_{t}- \hat{X}_{t}\vert^{m} \right]^{\frac{1}{m}}+1,
\end{align}
where the inequality follows from Minkowski's inequality and that $\mathbb{E}\left[ \vert \hat{X}_{t}-\hat{X}_{t}^{(c)}\vert^{m} \right]\leq 1$. If $m \in (0,1]$, 
\begin{align}
\label{rootinequalityforquantization}
&\mathbb{E}\left[ \vert X_{t}- \hat{X}_{t}^{(c)}\vert^{m} \right]\nonumber\\
\leq &   \mathbb{E}\left[ \left(\vert X_{t}- \hat{X}_{t}\vert+\vert \hat{X}_{t}-\hat{X}_{t}^{(c)}\vert\right)^{m} \right]\nonumber\\ 
\leq & \mathbb{E}\left[\vert X_{t}- \hat{X}_{t}\vert^{m}+\vert \hat{X}_{t}-\hat{X}_{t}^{(c)}\vert^{m} \right]\nonumber\\ 
\leq & \mathbb{E}\left[\vert X_{t}- \hat{X}_{t}\vert^{m}\right]+1, 
\end{align}
where the first inequality is due to triangle inequality, the second inequality follows from the inequality that $(a+b)^{m}\leq a^{m}+b^{m}$ for $a,b \geq 0$ when $m\in (0,1]$, and the third inequality is due to $\mathbb{E}\left[ \vert \hat{X}_{t}-\hat{X}_{t}^{(c)}\vert^{m} \right]\leq 1$.
 Hence, combining (\ref{Minkowskiinequalityforquantization}) and (\ref{rootinequalityforquantization}), we conclude that:
\begin{equation}
\label{trackingintegerforgeneral}
\sup_{t>0}\mathbb{E}\left[ \vert X_{t}- \hat{X}_{t}\vert^{m} \right]< \infty
\end{equation}
 holds only if
\begin{equation}
\label{trackingintegerforintegerestimator}
\sup_{t>0}\mathbb{E}\left[ \vert X_{t}- \hat{X}_{t}^{(c)}\vert^{m} \right]< \infty.
\end{equation}
Accordingly,  (\ref{trackingintegerforintegerestimator})  is a necessary condition to satisfy (\ref{trackingintegerforgeneral}).

Now, we find a necessary condition for (\ref{trackingintegerforintegerestimator}). Let $E_{t}:= X_{t}-\hat{X}_{t}^{(c)}$ be estimation error for estimators $\lbrace \hat{X}_{t}^{(c)} \rbrace$. As $ \vert X_{t} \vert \leq c_{t}$ for $t > 0$  and $\hat{X}_{t}^{(c)}$ is integer-valued, $E_{t}$ is an integer valued random variable taking values in $[-2c_{t},2c_{t}]$.

Using Lemma \ref{momententropyforboundedintegers} for $E_{t}$ being conditioned on $Y_{1:t}$, we have:
\begin{eqnarray}
\label{momententropyforconditionedE}
& &\mathbb{E}\left[ \vert E_{t} \vert^{m} \mid Y_{1:t}=y_{1:t}\right]+1 \geq  
\nonumber\\
& &\mathbb{E}\left[ \vert E_{t} \vert^{\rho} \mid Y_{1:t}=y_{1:t}\right]+1 \geq  \left( 3+2\log(2c_{t})\right)^{-\rho}
\nonumber\\
&\times &\mathbb{E}\left[ P_{E_{t}\vert Y_{1:t}}(E_{t}\vert Y_{1:t})^{-\frac{\rho}{\rho+1}}\mid Y_{1:t}=y_{1:t}\right]^{\rho+1}
\end{eqnarray}
where$P_{E_{t}\vert Y_{1:t}}$ is the conditional distribution for $E_{t}$ conditioned on $Y_{1:t}$ and the first inequality is due to that $E_{t}$ is integer-valued and the second inequality is due to Lemma \ref{momententropyforboundedintegers}.

As $(E_{t},Y_{1:t})\rightarrow (X_{t},Y_{1:t})$ is a bijective transformation when both $X_{t}$ and $\hat{X}_{t}^{(c)}$ are integer-valued, (\ref{momententropyforconditionedE}) becomes:
\begin{eqnarray}
\label{momententropyforconditionedX}
& &\mathbb{E}\left[ \vert E_{t} \vert^{m} \mid Y_{1:t}=y_{1:t}\right]+1 \geq  \left( 3+2\log(2c_{t})\right)^{-\rho}\nonumber\\
&\times &\mathbb{E}\left[ P_{X_{t}\vert Y_{1:t}}(X_{t}\vert Y_{1:t})^{-\frac{\rho}{\rho+1}}\mid Y_{1:t}=y_{1:t}\right]^{\rho+1}.
\end{eqnarray}
Taking expectations over $Y_{1:t}$ on both sides in (\ref{momententropyforconditionedX}) gives:
\begin{eqnarray}
\label{momententropyforoverall}
& &\mathbb{E}\left[ \vert E_{t} \vert^{m} \right]+1 \geq  \left( 3+2\log(2c_{t})\right)^{-\rho}\nonumber\\
&\times &\mathbb{E}\left[ \mathbb{E}\left[ P_{X_{t}\vert Y_{1:t}}(X_{t}\vert Y_{1:t})^{-\frac{\rho}{\rho+1}}\mid Y_{1:t}\right]^{\rho+1}\right] . 
\end{eqnarray}
Now, consider
\begin{eqnarray}
\label{reverseHolderforRenyipower}
&&\mathbb{E}\left[ P_{X_{t}\vert Y_{1:t}}(X_{t}\vert Y_{1:t})^{-\frac{\rho}{\rho+1}}\mid Y_{1:t}\right]^{\rho+1} \nonumber\\
& \geq &\mathbb{E}\left[ e^{-\frac{\rho}{p(\rho+1)}i(X_{t};Y_{1:t})} \mid Y_{1:t}\right]^{p(\rho+1)}\nonumber\\
& &\times\mathbb{E}\left[ P_{X_{t}}(X_{t})^{\frac{\rho}{(p-1)(\rho+1)}}\mid Y_{1:t}\right]^{(1-p)(\rho+1)},
\end{eqnarray}
where the inequality follows from the reverse H\"{o}lder inequality for $p \in (1, \infty)$ and $i(X_{t};Y_{1:t})$ is the information density for $P_{X_{t}Y_{1:t}}$. Then, we can get:
\begin{eqnarray}
\label{reverseHolderforRenyi}
&&\frac{1}{\rho}\log\mathbb{E}\left[\mathbb{E}\left[ P_{X_{t}\vert Y_{1:t}}(X_{t}\vert Y_{1:t})^{-\frac{\rho}{\rho+1}}\mid Y_{1:t}\right]^{\rho+1}\right]\nonumber\\ 
&\geq &  \frac{1}{\rho}\log\mathbb{E}\left[\mathbb{E}\left[ e^{-\frac{\rho}{p(\rho+1)}i(X_{t};Y_{1:t})} \mid Y_{1:t}\right]^{p(\rho+1)}\right]\nonumber\\
&&+H_{\alpha}(X_{t}),
\end{eqnarray}
where $\alpha=(p(\rho+1)-1)/((p-1)(\rho+1))$.

Combining (\ref{momententropyforoverall}) and (\ref{reverseHolderforRenyi}):
\begin{eqnarray}
\label{momententropyforoverallandlogm}
&&\frac{1}{\rho} \log\left( \mathbb{E}\left[ \vert E_{t} \vert^{m} \right]+1\right) \nonumber\\
& \geq &\frac{1}{\rho}\log\mathbb{E}\left[\mathbb{E}\left[ e^{-\frac{\rho}{p(\rho+1)}i(X_{t};Y_{1:t})} \mid Y_{1:t}\right]^{p(\rho+1)}\right]\nonumber\\
 && +  H_{\alpha}(X_{t})- \log \left(  3+2\log(c_{t})\right).
\end{eqnarray}
As $\lim_{t \rightarrow \infty}\log(\log(c_{t}))/t =0$,
\begin{equation}
\label{finitec}
\limsup_{t \rightarrow \infty} \frac{-1}{t}\log \left(  3+2\log(2c_{t})\right)=0.
\end{equation}
Therefore, combining (\ref{momententropyforoverallandlogm}) and (\ref{finitec}) implies that:
\begin{equation}
\label{necessaryconditionforlogepsilon}
\limsup_{t \rightarrow \infty} \frac{1}{\rho t} \log\left( \mathbb{E}\left[ \vert E_{t} \vert^{m} \right]+1\right) < \infty
\end{equation}
holds only \footnote{Here, it is possible that (\ref{necessaryconditionforlogepsilon}) holds when both limits in (\ref{necessaryconditionforGallagerRenyi:proof}) diverge. However, observe that the LHS of (\ref{necessaryconditionforGallagerRenyi:proof}) converges as $i(X_{t};Y_{1:t})$ is uniformly bounded by   $\log(\vert \mathcal{Y}\vert)$ and $\vert \mathcal{Y}\vert$ is finite. } if
\begin{eqnarray}
\label{necessaryconditionforGallagerRenyi:proof}
&&\liminf_{t \rightarrow \infty}-\frac{1}{\rho t}\log\mathbb{E}\left[\mathbb{E}\left[ e^{-\frac{\rho}{p(\rho+1)}i(X_{t};Y_{1:t})} \mid Y_{1:t}\right]^{p(\rho+1)}\right] 
\nonumber\\ 
&\!\! \!\! \geq \!\! \!\!&\limsup_{t \rightarrow \infty} \frac{1}{t}H_{\alpha}(X_{t}).
\end{eqnarray}
In addition, if (\ref{trackingintegerforintegerestimator}) holds then (\ref{necessaryconditionforlogepsilon})  holds. Hence (\ref{necessaryconditionforGallagerRenyi:proof}) is a necessary condition for (\ref{trackingintegerforintegerestimator}). Therefore, (\ref{necessaryconditionforGallagerRenyi:proof}) is a necessary condition for (\ref{trackingintegerforgeneral}), i.e., $\lbrace X_{t}\rbrace$ being order $m$ moment trackable though process $\lbrace Y_{t} \rbrace$. As $p>1$ is arbitrary, $p(\rho+1)$ can be replaced with an arbitrary $q$ such that $q>\rho+1$.   

\subsection{The Proof of Lemma \ref{momententropyforboundedintegers}}
\label{momententropyforboundedintegers:proof}
 

We have two methods to prove Lemma \ref{momententropyforboundedintegers}. The first method follows  the proof  of Theorem 1 in \cite{Arikan1996}, with the guessing function replaced by $A(x)$  defined in \eqref{epsilonA} below and some other necessary changes. In the sequel, we provide a second proof method, which is based on the reverse H\"{o}lder inequality approach used in \cite[Lemma 1]{Arikan1996} and  in \cite[Theorem 2.1]{li2018large}.

Let us define the following function:
\begin{equation}
\label{epsilonA}
A(x) = \begin{cases} \vert x \vert &\mbox{if } x \neq 0,\\
\epsilon & \mbox{if } x = 0, \end{cases}
\end{equation}
where $\epsilon$ is an arbitrary positive real number. Accordingly, observe that
\begin{align}
\label{absbound}
&\mathbb{E}[\vert X \vert^{\rho}]+\epsilon^{\rho} P_{X}(0) \nonumber\\
=& \displaystyle\sum_{x\in \mathcal{X}}P_{X}(x)A(x)^{\rho}
\nonumber\\
\geq&\left[\displaystyle\sum_{x\in \mathcal{X}} P_{X}(x)^{\frac{1}{p}}\right]^{p}\left[\displaystyle\sum_{x\in \mathcal{X}} A(x)^{\frac{-\rho}{p-1}}\right]^{-(p-1)},
\end{align}
where the inequality is due to the reverse H\"{o}lder inequality for $p \in (1,\infty)$.  Considering $p=1+\rho$ in (\ref{absbound}), we get   
\begin{align}
\label{absboundforrho}
&\mathbb{E}[\vert X \vert^{\rho}]+\epsilon^{\rho}  P_{X}(0)  
\nonumber\\
\geq&\left[\displaystyle\sum_{x\in \mathcal{X}} P_{X}(x)^{\frac{1}{1+\rho}}\right]^{1+\rho}\left[\displaystyle\sum_{x\in \mathcal{X}} A(x)^{-1}\right]^{-\rho}
\nonumber\\
\geq&\left[\displaystyle\sum_{x\in \mathcal{X}} P_{X}(x)^{\frac{1}{1+\rho}}\right]^{1+\rho}
\!\!\left[2 + \frac{1}{\epsilon}+\log(M_{-}M_{+})\right]^{-\rho}, 
\end{align}
where the second inequality is due to 
\begin{eqnarray}
\label{Abound}
\displaystyle\sum_{x=-M_{-}}^{M_{+}} A(x)^{-1} &=& \epsilon^{-1} +\displaystyle\sum_{i=1}^{M_{-}} i^{-1}+\displaystyle\sum_{j=1}^{M_{+}} j^{-1}\nonumber\\
&\leq & 2+\epsilon^{-1}  + \log(M_{-}M_{+}).\nonumber
\end{eqnarray}
Letting $\epsilon=1$, combining (\ref{absboundforrho}) with $P_{X}(0) \leq 1$ and 
\begin{align}
e^{\rho H_{\frac{1}{1+\rho}}(X)}=\left[\displaystyle\sum_{x\in \mathcal{X}} P_{X}(x)^{\frac{1}{1+\rho}}\right]^{1+\rho}, 
\end{align}
 we obtain (\ref{absmomentandRenyi}).

\subsection{The Proof of Lemma \ref{momententropyforintegers}}
\label{momententropyforintegers:proof}
The proof is similar to the proof of Lemma \ref{momententropyforboundedintegers} where $A(x)$ is defined as in (\ref{epsilonA}). We have
\begin{eqnarray}
\label{absboundforrhoandm}
&&\mathbb{E}[\vert X \vert^{m}]+1
\nonumber\\
&\geq &\left[\displaystyle\sum_{x=-\infty}^{\infty} P_{X}(x)^{\frac{1}{1+\rho}}\right]^{1+\rho}\left[\displaystyle\sum_{x=-\infty}^{\infty} A(x)^{-\frac{m}{\rho}}\right]^{-\rho}
\nonumber\\
&\geq &\left[\displaystyle\sum_{x=-\infty}^{\infty} P_{X}(x)^{\frac{1}{1+\rho}}\right]^{1+\rho}
\left[1+2\displaystyle\sum_{n=1}^{\infty} \frac{1}{n^{\frac{m}{\rho}}}\right]^{-\rho},
\end{eqnarray}
where the first inequality follows from reverse H\"{o}lder inequality and the second inequality follows from the choice of $\epsilon=1$.

\subsection{The Proof of Lemma \ref{MAPboundforgeneraldistance}}
\label{MAPboundforgeneraldistance:proof}
 We first consider the following:
\begin{equation}
\label{generaldistanceforMAPerror}
\mathbb{E}\left[ d(X_{t},\hat{X}_{t}^{(\text{MAP})} )\right]= \sum_{r=1}^{\infty}\Pr\left( d(X_{t},\hat{X}_{t}^{(\text{MAP})}) \geq r \right).
\end{equation}
Then, observe that:
\begin{eqnarray}
\label{doubleindicatorsforboundingrerror}
&&\Pr\left( d(X_{t},\hat{X}_{t}^{(\text{MAP})})\geq r \mid X_{t}=x, Y_{1:t}=y_{1:t}\right) \nonumber\\
&\leq &  \min \left\lbrace 1, \!\!\displaystyle\sum_{x':d(x,x')\geq r}\!\!\mathbb{I}_{\left\lbrace P_{X_{t}\vert Y_{1:t}}(x\vert y_{1:t}) \leq P_{X_{t}\vert Y_{1:t}}(x'\vert y_{1:t})\right\rbrace }\right\rbrace \nonumber\\
&\leq & \left( \sum_{x'}\mathbb{I}_{\left\lbrace d(x,x')\geq r\right\rbrace}\mathbb{I}_{\left\lbrace P_{X_{t}\vert Y_{1:t}}(x\vert y_{1:t}) \leq P_{X_{t}\vert Y_{1:t}}(x'\vert y_{1:t})\right\rbrace } \right)^{\rho}\nonumber\\
& \leq &\left( \sum_{x'}\left(\frac{d(x,x')}{r}\right)^{\frac{s}{\rho}}\left( \frac{P_{X_{t}\vert Y_{1:t}}(x'\vert y_{1:t})}{P_{X_{t}\vert Y_{1:t}}(x\vert y_{1:t})}\right)^{\frac{1}{\rho + 1}}\right)^{\rho}
\end{eqnarray}
where $\mathbb{I}_{\left\lbrace \cdot \right\rbrace }$ is the indicator function, the first inequality is due to the definition of MAP estimators, the second inequality follows from the inequality $\min\left\lbrace 1, x\right\rbrace \leq x^{\rho}$  for any $\rho \geq 0$ when $x$ is either $0$ or larger than equal to $1$, and the third inequality follows from the inequality $\mathbb{I}_{\left\lbrace x \geq r \right\rbrace } \leq (x/r)^{s}$ for any $x\geq 0,r>1,s \geq 0$.
Combining (\ref{doubleindicatorsforboundingrerror}) and (\ref{generaldistanceforMAPerror}) gives (\ref{upperboundforMAPerrorongeneraldistance}).


\subsection{The Proof of Theorem \ref{sufficientconditionforrhoestimators}}
\label{proof:sufficientconditionforrhoestimators}
Consider $\hat{X}_{t}^{(\rho)}$ for $\rho=sm(m+1)$, and the following:
\begin{equation}
\label{mmomentforrhoerror}
\mathbb{E}\left[ \vert X_{t}-\hat{X}_{t}^{(\rho)}\vert^{m}\right] = \sum_{r=1}^{\infty}\Pr\left( \vert X_{t}-\hat{X}_{t}^{(\rho)} \vert^{m} \geq r \right).
\end{equation}
Then observe that:
\begin{eqnarray}
\label{doubleindicatorsforboundingrhorerror}
&&\Pr\left( \vert X_{t}-\hat{X}_{t}^{(\rho)} \vert^{m}\geq r \mid X_{t}=x, Y_{1:t}=y_{1:t}\right) \nonumber\\
&\leq &  \min \left\lbrace 1, \!\!\displaystyle\sum_{x':\vert x-x \vert^{m}\geq r}\!\!\mathbb{I}_{\left\lbrace Q(x,x')\leq P_{X_{t}\vert Y_{1:t}}(x'\vert y_{1:t})\right\rbrace }\right\rbrace \nonumber\\
&\leq &  \left( \sum_{x'}\mathbb{I}_{\left\lbrace \vert x-x \vert^{m}\geq r\right\rbrace}\mathbb{I}_{\left\lbrace Q(x,x') \leq P_{X_{t}\vert Y_{1:t}}(x'\vert y_{1:t})\right\rbrace } \right)^{\rho}\nonumber\\
& \leq &\left[ \sum_{x'\neq x}\left(\frac{\vert x-x \vert^{m}}{r}\right)^{s}\left( \frac{P_{X_{t}\vert Y_{1:t}}(x'\vert y_{1:t})}{Q(x,x')}\right)^{\frac{1}{m + 1}}\right]^{m}, \nonumber\\
\end{eqnarray}
where $Q(x,x')=J_{t}(\rho,y_{1:t})\vert x-x' \vert^{\rho} P_{X_{t}\vert Y_{1:t}}(x\vert y_{1:t})$  
, $\mathbb{I}_{\left\lbrace \cdot \right\rbrace }$ is the indicator function, and the inequalities follow similarly as (\ref{doubleindicatorsforboundingrerror}).
Combining (\ref{doubleindicatorsforboundingrhorerror}) and (\ref{mmomentforrhoerror}) gives:
\begin{eqnarray}
\label{mmomentforrhoerrorforJt}
&&\mathbb{E}\left[ \vert X_{t}-\hat{X}_{t}^{(\rho)} \vert^{m}\right]\leq \zeta(s)\times \nonumber\\
&&\mathbb{E}\left[ \left[\frac{1}{J_{t}(\rho,Y_{1:t})}\right]\mathbb{E}\left[ P_{X_{t}\vert Y_{1:t}}(X_{t}\vert Y_{1:t})^{-\frac{m}{m+1}}\mid  Y_{1:t}\right]^{m+1}\right].\nonumber\\
\end{eqnarray}
Applying H\"{o}lder inequality to the RHS of (\ref{mmomentforrhoerrorforJt}) gives the sufficient condition.

\subsection{The Proof of (\ref{GartnerEllisfortrackinginformationdensity})}
\label{proof:GartnerEllisfortrackinginformationdensity}
Consider (\ref{necessaryconditionforGallagerRenyicorollary2})  for a rate-$R$ source for $S_{t}$:
\begin{equation}
\label{rateRforGallager}
\liminf_{t \rightarrow \infty}\frac{1}{\rho t}-\log\mathbb{E}\left[\mathbb{E}\left[  e^{-\frac{\rho}{1+\rho}i(X_{1:t};Y_{1:t})} \Big| Y\right]^{1+\rho}\right] \geq R\log(2).
\end{equation}
Then, we have:
\begin{align}
\label{twiceJensenforGallager}
&\mathbb{E}\left[\mathbb{E}\left[  e^{-\frac{\rho}{1+\rho}i(X_{1:t};Y_{1:t})} \Big| Y\right]^{1+\rho}\right]^{-1}\nonumber\\
\leq & \mathbb{E}\left[\mathbb{E}\left[  e^{\rho i(X_{1:t};Y_{1:t})} \Big| Y\right]^{-1}\right]^{-1}\nonumber\\
\leq & \mathbb{E}\left[\mathbb{E}\left[  e^{\rho i(X_{1:t};Y_{1:t})} \Big| Y\right]\right],
\end{align}
where the inequalities follow from Jensen's inequality. Combining (\ref{twiceJensenforGallager}) with (\ref{rateRforGallager}) gives (\ref{GartnerEllisfortrackinginformationdensity}).

\bibliography{AgeOfInformation}
\end{document}